\documentclass[12pt,oneside]{amsart}

\usepackage{amsthm,amsmath,amssymb}
\usepackage{mathrsfs}
\usepackage{enumerate}
\usepackage{amsfonts}
\usepackage{verbatim}
\usepackage{amsbsy}
\usepackage{cite}
\usepackage{amsmath}
\usepackage{amssymb}
\usepackage{MnSymbol}
\usepackage{mathrsfs}

\newtheorem{theorem}{Theorem}[section]
\newtheorem{lemma}[theorem]{Lemma}
\newtheorem{proposition}[theorem]{Proposition}
\newtheorem{conjecture}[theorem]{Conjecture}

\newtheorem{question}[theorem]{Question}
\newtheorem{corollary}[theorem]{Corollary} 
\theoremstyle{definition}
\newtheorem{definition}[theorem]{Definition}

\theoremstyle{remark}
\newtheorem{remark}[theorem]{Remark}
\newtheorem{example}[theorem]{Example}

\DeclareMathOperator{\dist}{dist}


\begin{document}

\title{Parikh matrices and Parikh Rewriting Systems}
\author{Wen Chean Teh}
\address{School of Mathematical Sciences\\
Universiti Sains Malaysia\\
11800 USM\\
Malaysia}
\email{dasmenteh@usm.my}

\maketitle

\begin{abstract}
Since the introduction of the Parikh matrix mapping, its injectivity problem is on top of the list of open problems in this topic. In 2010 Salomaa provided a solution for the ternary alphabet in terms of a Thue system with an additional feature called counter. This paper proposes the notion of a Parikh rewriting system as a generalization and systematization of Salomaa's result. It will be shown that every Parikh rewriting system induces a Thue system without counters that serves as a feasible solution to the injectivity problem.
\end{abstract}

\section{Introduction}

The Parikh matrix mapping \cite{MSSY01} was originally introduced as an extension of the Parikh mapping \cite{rP66}. 
Parikh matrices provide more structural information about words than Parikh vectors and is a useful tool in studying subword occurrences. Two words are $M$-equivalent if{f} they have the same Parikh matrix.
The characterization of $M$-equivalence, also known as the injectivity problem, and the closely related ambiguity problem have garnered the most interest among researchers in this area \cite{aA07,aA10,aA14,AAP08,AMM01a,FR04,MS12,MS04,aS05a,aS10, SY10,vS09,SS06,wT14,wT15,TK14}.
Complete characterization for the case of the binary alphabet was obtained \cite{aA07,FR04}.
However, the case of the ternary alphabet proves to be elusive.

The first complete characterization was provided by Salomaa \cite{aS10} in terms of a Thue system. Nevertheless, the associated rewriting rules do not preserve \mbox{$M$-equivalence}. To compensate, a counter is attached to the Thue system to keep track of the quantity that decides whether the resulting word is \mbox{$M$-equivalent} to the original word. 
Later Atanasiu \cite{aA14} proposed a wider class of $M$-equivalence preserving rewriting rules for the ternary alphabet, based on Salomaa's Thue system. 

This paper introduces the concept of a Parikh rewriting system. It generalizes the above Thue system proposed by Salomaa. 
Then it will be shown that
``irreducible" transformations extracted from a Parikh rewriting system induces a Thue system that  characterizes $M$-equivalence.

The remainder of this paper is structured as follows.  Section~2 provides the basic terminology and preliminary.
Section 3 systematizes the characterization of \mbox{$M$-equivalence} in terms of Thue systems.
The main object of our study, namely the Parikh rewriting system, will be introduced in Section~4.
Irreducible transformations of a Parikh rewriting system will be studied in the subsequent section.
The next two sections involve some case studies for the binary and the ternary alphabet. Our conclusions follow after that.


\section{Parikh Matrices}

We will assume the reader is familiar with the basic notions of formal languages.
The reader is referred to \cite{RS97} if necessary.

Suppose $\Sigma$ is a finite alphabet. The set of words over $\Sigma$ is denoted by $\Sigma^*$. The empty word is denoted by $\lambda$. 
Let $\Sigma^+$ denote the set of nonempty words over $\Sigma$. If $v,w\in \Sigma^*$, the concatenation of $v$ and $w$ is denoted by  $vw$.
An \emph{ordered alphabet} is an alphabet $\Sigma= \{a_1, a_2, \dotsc,a_s\}$ with a total ordering on it. For example, if $a_1<a_2<\dotsb < a_s$, then we may write
$\Sigma= \{a_1<a_2< \dotsb<a_s\}$. 
On the other hand, if $ \Sigma=\{a_1< a_2< \dotsb< a_s\}  $ is an ordered alphabet, then the \emph{underlying alphabet} is  $\{a_1, a_2, \dotsc,a_s\}$.
Frequently, we will abuse notation and use $\Sigma$ to stand for both the ordered alphabet and its underlying alphabet, for example, as in ``$w\in \Sigma^*$" when $\Sigma$ is an ordered alphabet. 
If $w\in \Sigma^*$, then $\vert w\vert$ is the length of $w$  and $w[i]$ is the $i$-th letter of $w$.
For $1\leq i\leq j \leq s$, let $a_{i,j}$ denote the word $a_ia_{i+1}\dotsm a_j$.
Suppose $\Gamma\subseteq \Sigma$. The projective morphism $\pi_{\Gamma}\colon \Sigma^*\rightarrow \Gamma^*$ is defined by
$$\pi_{\Gamma}(a)=\begin{cases}
a, & \text{if } a\in \Gamma\\
\lambda, & \text{otherwise.}
\end{cases}$$
We may write $\pi_{a,b}$ for $\pi_{\{a,b\}}$.

\begin{definition}
A word $w'$ is a \emph{subword} of $w\in \Sigma^*$ if{f} there exist $x_1,x_2,\dotsc, x_n$, $y_0, y_1, \dotsc,y_n\in \Sigma^*$, some of them possibly empty, such that
$$w'=x_1x_2\dotsm x_n \text{ and } w=y_0x_1y_1\dotsm y_{n-1}x_ny_n.$$
\end{definition}

In the literature,  our subwords are usually called ``scattered subwords". A \emph{factor} is a contiguous subword. The number of occurrences of a word $u$ as a subword of $w$ is denoted by $\vert w\vert_u$. 
Two occurrences of $u$ are considered different if{f} they differ by at least one position of some letter. 
For example, $\vert aabab\vert_{ab}=5$ and $\vert baacbc\vert_{abc}=2$.
By convention, $\vert w\vert_{\lambda}=1$ for all $w\in \Sigma^*$.

For any integer $k\geq 2$, let $\mathcal{M}_k$ denote the multiplicative monoid of $k \times k$ upper triangular matrices with nonnegative integral entries and unit diagonal.

\begin{definition} 
Suppose $\Sigma=\{a_1<a_2< \dotsb<a_s\}$ is an ordered alphabet. The \emph{Parikh matrix mapping}, denoted $\Psi_{\Sigma}$, is the monoid morphism
$$ \Psi_{\Sigma}\colon \Sigma^*\rightarrow \mathcal{M}_{s+1}$$
defined as follows:\\
if $\Psi_{\Sigma}(a_q)=(m_{i,j})_{1\leq i,j\leq s+1}$, then $m_{i,i}=1$ for each $1\leq i\leq s+1$, $m_{q,q+1}=1$ and all other entries of the matrix $\Psi_{\Sigma}(a_q)$ are zero. 
Matrices of the form  $\Psi_{\Sigma}(w)$ for $w\in \Sigma^*$ are called \emph{Parikh matrices}.
\end{definition}


\begin{theorem}\cite{MSSY01}
Suppose $\Sigma=\{a_1<a_2< \dotsb<a_s\}$ is an ordered alphabet and $w\in \Sigma^*$. The matrix $\Psi_{\Sigma}(w)=(m_{i,j})_{1\leq i,j\leq s+1}$ has the following properties:
\begin{itemize}
\item $m_{i,i}=1$ for each $1\leq i \leq s+1$;
\item $m_{i,j}=0$ for each $1\leq j<i\leq s+1$;
\item $m_{i,j+1}=\vert w \vert_{a_{i,j}}$ for each $1\leq i\leq j \leq s$.
\end{itemize}
\end{theorem}

The \emph{Parikh vector} $\Psi(w)=(\vert w\vert_{a_1}, \vert w\vert_{a_2}, \dotsc, \vert w\vert_{a_s})$ of a word $w\in \Sigma^*$ is contained in the second diagonal\footnote{The \emph{second diagonal} of a matrix in $\mathcal{M}_{k+1}$ is the diagonal of length $k$ immediately above the principal diagonal.}  of the Parikh matrix $\Psi_{\Sigma}(w)$.

\begin{example}
Suppose $\Sigma=\{a<b<c\}$ and $w=abcbac$.
Then \begin{align*}
\Psi_{\Sigma}(w)&=\Psi_{\Sigma}(a)\Psi_{\Sigma}(b)\Psi_{\Sigma}(c)\Psi_{\Sigma}(b)\Psi_{\Sigma}(a)\Psi_{\Sigma}(c)\\
&= \begin{pmatrix}
1 & 1& 0&0 \\
0 &1 & 0 & 0\\
0 & 0 & 1 & 0\\
0 & 0 & 0 & 1
\end{pmatrix}
\begin{pmatrix}
1 & 0 & 0&0 \\
0 &1 & 1 & 0\\
0 & 0 & 1 & 0\\
0 & 0 & 0 & 1
\end{pmatrix}\dotsm
\begin{pmatrix}
1 & 0 & 0&0 \\
0 &1 & 0 & 0\\
0 & 0 & 1 & 1\\
0 & 0 & 0 & 1
\end{pmatrix}\\
&= \begin{pmatrix}
1 & 2& 2& 3 \\
0 &1 & 2 & 3\\
0 & 0 & 1 & 2\\
0 & 0 & 0 & 1
\end{pmatrix}
=\begin{pmatrix}
1 & \vert w\vert_a & \vert w\vert_{ab} & \vert w\vert_{abc} \\
0 &1 & \vert w\vert_b & \vert w\vert_{bc}\\
0 & 0 & 1 & \vert w\vert_c\\
0 & 0 & 0 & 1
\end{pmatrix}.
\end{align*}
\end{example}

\begin{definition}
Suppose $\Sigma=\{a_1<a_2< \dotsb<a_s\}$ is an ordered alphabet.
Two words $w,w'\in \Sigma^*$ are \emph{$M$-equivalent}, denoted $w\equiv_M w'$, if{f} $\Psi_{\Sigma}(w)=\Psi_{\Sigma}(w')$.
A word $w\in \Sigma^*$ is \emph{$M$-ambiguous} if{f} there exists some distinct word $w'\in \Sigma^*$ $M$-equivalent to $w$. Otherwise, $w$ is said to be \emph{$M$-unambiguous}.
\end{definition}

\section{Thue Systems in Relation to $M$-Equivalences}

Characterization of $M$-equivalences in terms of Thue systems has been carried out in \cite{aA07,AAP08,AMM01a,FR04,MS04,aS10}.
However, a systematic treatment of this has not appeared in the literature.
To us a Thue system is defined as follows.

\begin{definition}
A \emph{Thue system} is an ordered pair $(\Sigma, R)$, where $\Sigma$ is an ordered alphabet\footnote{Although unnecessary in this definition, the alphabet is assumed to be ordered out of convenience and compatibility with the next definition.} and $\emptyset \neq R\subseteq \{ \,(y,y') \mid y,y' \in \Sigma^* \text{ and } y\neq y'\,\}$. 
Each element $(y,y')$ of $R$ is associated with two \emph{rewriting rules} of the Thue system, 
which are suggestively denoted by $y\rightarrow y'$ and $y'\rightarrow y$. 
Suppose $w,w'\in \Sigma^*$. We say that $w$ \emph{transforms directly into} $w'$, denoted $w\rightarrow_R w'$, using the rewriting rule $y\rightarrow y'$ if{f} 
$$w=xyz  \text{ and } w'=xy'z \text{ for some } x,z \in \Sigma^*.$$ 
We say $w$ \emph{transforms into} $w'$, denoted $w\Rightarrow_{R} w'$, if{f} 
there exist $w_0,w_1, \dotsc,w_n\in \Sigma^*$ such that
$$w=w_0 \rightarrow_R w_1 \rightarrow_R w_2 \rightarrow_R \dotsb \rightarrow_R w_n=w'.$$
By definition, $w \Rightarrow_{R} w$ for every $w\in \Sigma^*$.
\end{definition}

Note that $\Rightarrow_R$ is an equivalence relation on $\Sigma^*$. The corresponding equivalence class containing the word $w\in \Sigma^*$
will be denoted by $[w]_R$.



\begin{definition}
A Thue system $(\Sigma,R)$ 
is \emph{Parikh sound}
if{f} $w\equiv_M w'$ for every $w, w'\in \Sigma^*$ such that $w\Rightarrow_{R} w'$.
A Thue system $(\Sigma,R)$ is \emph{Parikh complete} 
if{f} $w\Rightarrow_R w'$ for every $w, w'\in \Sigma^*$ such that $w\equiv_M w'$.
\end{definition}

A Thue system $(\Sigma,R)$ is Parikh sound and Parikh complete if and only if
the relation $\Rightarrow_R$ coincides with the $M$-equivalence relation $\equiv_M$.

\begin{proposition}\label{0403b}
Suppose  $(\Sigma,R)$ is a Thue system.
If $y\equiv_M y'$ for every $(y,y')\in R$, then $(\Sigma,R)$ is Parikh sound.
\end{proposition}

\begin{proof}
The proof is straightforward using the left-invariance, right-invariance, and transitivity of $M$-equivalence.
\end{proof}

\begin{example}
Suppose $\Sigma$ is any ordered alphabet and $R=\{\,   (w,w')\in \Sigma^*\times \Sigma^* \mid w \equiv_M w' \text{ and } w\neq w' \,\}$.
Then the Thue system $( \Sigma,   R     )$ is trivially Parikh sound and Parikh complete.
\end{example}

\begin{example}\cite{aA07,FR04}\label{1506b}
Suppose $\Sigma=\{a<b\}$ and $R=\{ \, (abxba,baxab) \mid x\in \Sigma^*\,\}$.
Then the Thue system $( \Sigma,   R     )$ is Parikh sound and Parikh complete.
\end{example}

\begin{example}\label{0701c}
Suppose $\Sigma=\{a<b<c\}$ and 
$$R=\{ (ac,ca)\}  \cup \{ \, (abxba, baxab) \mid x\in \{a,b\}^*\,\}\cup\{\, (bcxcb, cbxbc) \mid x\in \{b,c\}^*   \, \} .$$
Then the Thue system $( \Sigma, R )$ is Parikh sound but not Parikh complete.
The two words $abbcbacb$ and $bacbabbc$
are $M$-equivalent but $abbcbacb \nRightarrow_R bacbabbc$.

This Thue system $(\Sigma,R)$ was studied in \cite{AAP08}.
The fact that it is not Parikh complete was pointed out in \cite{vS09}, where the counterexample provided is the pair of words
$babcbabcbabcbab$ and $bbacabbcabbcbba$, each with \emph{fifteen} letters.
\end{example}


\begin{theorem}\label{1228d}
Suppose $\Sigma$ is an ordered alphabet with $\vert \Sigma\vert \geq 2$. No Parikh sound Thue system $(\Sigma, R)$ with $R$ finite is Parikh complete.
\end{theorem}

\begin{proof}
Suppose $\Sigma=\{a<b<\dotsb\}$ and $(\Sigma,R)$ is Parikh sound with $R$ finite. Let $N$ be an integer at least the length of every word appearing in some rewriting rule of the Thue system. Let $w=ab\underbrace{a\dotsm a}_{N-1 \text{times}}ba$ and $w'= ba\underbrace{a\dotsm a}_{N-1 \text{times}}ab$.
Then $w$ and $w'$ are $M$-equivalent. However, every factor of $w$ with length at most $N$ is $M$-unambiguous. Since $(\Sigma,R)$ is Parikh sound, every word appearing in its rewriting rules is $M$-ambiguous. Hence, no rewriting rule
of the Thue system can be applied to $w$.
Therefore, $w\nRightarrow_R w'$ and thus $(\Sigma,R)$ is not Parikh complete.
\end{proof}

\section{Parikh Rewriting Systems}

Now, we introduce our main object of study.

\begin{definition}
A \emph{Parikh rewriting system} is a triple $\mathfrak{P}=(\Sigma,R, C )$, where $\Sigma=\{a_1<a_2<\dotsb<a_s\}$ is an ordered alphabet of size at least two, $(\Sigma,R)$ is a Thue system such that each rewriting rule preserves the Parikh vector, and $C\subseteq \{\, a_{i,j} \mid   1\leq i< j \leq s\, \}$. The elements of $C$ are called the \emph{counters}.
Suppose $w, w'\in \Sigma^*$. We say that $w$ \emph{transforms into} $w'$, denoted $w\Rightarrow_{\mathfrak{P}} w'$, if{f} 
$\vert w\vert_c =\vert w'\vert_c$ for all $c\in C$ and $w\Rightarrow_R w'$.
Two Parikh rewriting systems $\mathfrak{P}$ and $\mathfrak{Q}$ are \emph{equivalent} if{f} $\Rightarrow_{\mathfrak{P}}=\Rightarrow_{\mathfrak{Q}}$.
\end{definition}

Note that $\Rightarrow_{\mathfrak{P}}$ is an equivalence relation on $\Sigma^*$.

\begin{definition}
Suppose $\mathfrak{P}=(\Sigma, R, C)$ is a Parikh rewriting system.
\begin{enumerate}
\item Two words $w,w'\in \Sigma^*$ are \emph{$\mathfrak{P}$-equivalent} if{f} $w\Rightarrow_{\mathfrak{P}} w'$.
\item A word $w\in \Sigma^*$ is \emph{$\mathfrak{P}$-ambiguous} if{f} there exists some distinct word $w'\in \Sigma^*$ $\mathfrak{P}$-equivalent to $w$. Otherwise, $w$ is said to be \emph{$\mathfrak{P}$-unambiguous}.
\end{enumerate}
\end{definition}

It is not absolutely obvious that $\mathfrak{P}$-ambiguous words must exist for every Parikh rewriting system $\mathfrak{P}$.

\begin{lemma}\label{1611a}
Suppose $\Sigma$ is any alphabet and suppose $w_1,w_1'\in \Sigma^*$ are distinct and $\Psi(w_1)=\Psi(w_1')$.
Recursively, define $w_{n+1}=w_nw_n'$ and $w_{n+1}'=w_n'w_n$. Then for every positive integer $n$, the words $w_n$ and $w_n'$ are distinct and $\vert w_n\vert_u=\vert w_n'\vert_u$ for every $u\in \Sigma^*$ such that $\vert u\vert \leq n$. 
\end{lemma}

\begin{remark}
Lemma~\ref{1611a} can be proved easily by induction on $n$. 
The above lemma is implicit in the proof of $\xi(t)<2^t$ that appeared in \cite{MMSSS91}, where $\xi(t)$ is  the maximal length such that any word of length at most $\xi(t)$ over any alphabet of size at least two  is uniquely determined by its $t$-spectrum.
\end{remark}

\begin{proposition}\label{1501a}
$\mathfrak{P}$-ambiguous words exist for any  Parikh rewriting system  $\mathfrak{P}=(\Sigma, R, C )$.
\end{proposition}

\begin{proof}
Recall that every rewriting rule of $\mathfrak{P}$ preserves the Parikh vector.  Suppose $(w,w')\in R$. By our definition, $w\neq w'$. Let $w_1=w$ and $w_1'=w'$. For every $2\leq n\leq \vert \Sigma\vert$, define 
$w_n$ and $w_n'$  as in Lemma~\ref{1611a}. 
We claim that $w_{\vert \Sigma \vert} \Rightarrow_{\mathfrak{P}} w_{\vert \Sigma\vert}'$.
Since $\vert u\vert \leq \vert \Sigma\vert$ for every $u\in C$, by Lemma~\ref{1611a}, it suffices to show that $w_{\vert \Sigma\vert} \Rightarrow_R w_{\vert \Sigma\vert}'$.
This can be proved easily by induction on $n$. Clearly, $w=w_1 \Rightarrow_R w_1'=w'$.
Assume $w_n \Rightarrow_R w_n'$. By symmetry,  $w_n' \Rightarrow_R w_n$. Therefore, 
$w_{n+1}=w_nw_n' \Rightarrow_R w_n'w_n' \Rightarrow_R w_n'w_n=w_{n+1}'$ as required.
\end{proof}

\begin{definition}
Suppose $\mathfrak{P}=(\Sigma,R, C )$ is a Parikh rewriting system.
We say that $\mathfrak{P}$ is \emph{sound} if{f} $w\equiv_M w'$  for every $w, w'\in \Sigma^*$ such that $w\Rightarrow_{\mathfrak{P}} w'$.
We say that $\mathfrak{P}$ is \emph{complete}
if{f} $w\Rightarrow_{\mathfrak{P}} w'$ for every $w, w'\in \Sigma^*$ such that $w\equiv_M w'$.
\end{definition}

Thus a Parikh rewriting system $\mathfrak{P}$ is sound and complete if and only if $\Rightarrow_{\mathfrak{P}}=\equiv_M$.

\begin{remark}\label{1001b}
If $(\Sigma, R,C)$ is a complete Parikh rewriting system, then $(\Sigma,R)$ is a Parikh complete Thue system.
\end{remark}

\begin{example}\label{1506a}
Suppose $\Sigma=\{a<b\}$ and $R=\{  (ab,ba) \}$.
Then the Parikh rewriting system $(\Sigma, R ,\{ab\})$ is sound and complete. This was Lemma~4 in \cite{MS04}, which appeared equivalently in the following form.
\begin{quote}
Suppose $\Sigma=\{a<b\}$ and $w,w'\in \Sigma^*$. Then $w\equiv_M w'$ if and only if there exist $w_0,w_1,\dotsc, w_n\in \Sigma^*$ such that 
$w=w_0 \rightarrow_R w_1 \rightarrow_R w_2 \rightarrow_R \dotsb \rightarrow_R w_n=w'$ and the number of applications of the rewriting rule  
$ab\rightarrow ba$ equals the number of applications of the rewriting rule $ba\rightarrow ab$.
\end{quote}
\end{example}

\begin{example}\label{2111a}
Suppose $\Sigma=\{a<b<c\}$. 
Let $$R=\{ (ac,ca)\}  \cup \{ \, (abxba, baxab) \mid x\in \Sigma^*\,\}\cup\{\, (bcxcb, cbxbc) \mid x\in \Sigma^*   \, \} .$$
Then the Parikh rewriting system $( \Sigma,   R   ,\{abc\}     )$ is sound and complete.
This is Theorem~9 in \cite{aS10} and represents the first exhaustive solution to the injectivity problem for the ternary alphabet as mentioned in the Introduction. However, the rewriting rules do not preserve $M$-equivalence. This system of Salomaa will be studied in Section~\ref{2703a}.
\end{example}

\begin{example}\label{0901a}
Suppose $\Sigma=\{a<b<c\}$ and $R=\{ (ac, ca), 
 (bc,cb), (ab, ba)\}$. Apparently, $w\Rightarrow_R w'$ for every $w,w'\in\Sigma^*$ with $\Psi(w)=\Psi(w')$. Therefore, the Parikh rewriting system $( \Sigma,   R  ,\{ab,bc, abc\}     )$ is trivially sound and complete.
\end{example}

Example~\ref{0901a} suggests the following question: if finitely many rewriting rules are allowed, is there a sound and complete Parikh rewriting system for the ternary alphabet not involving all three counters? 

\begin{theorem}\label{2912a}
Suppose $\Sigma=\{a<b<c\}$ and $\mathfrak{P}=(\Sigma, R,C)$ is a sound Parikh rewriting system. If $(\Sigma,R)$ is not Parikh sound, then $abc\in C$.
\end{theorem}

\begin{proof}
We argue by contradiction. Assume $abc\notin C$ and so $C\subseteq \{ab,bc\}$. Since $(\Sigma,R)$ is not Parikh sound, by Proposition~\ref{0403b}, choose $(w,w')\in R$ such that
$w\nequiv_M w'$. Since every rewriting rule of a Parikh rewriting system preserves the Parikh vector, it follows that
$\vert w\vert_u\neq \vert w'\vert_u$ for some $u\in \{ab,bc,abc\}$. If $\vert w\vert_{ab}= \vert w'\vert_{ab}$ 
and $\vert w\vert_{ba}= \vert w'\vert_{ba}$, then $w\Rightarrow_{\mathfrak{P}} w'$ but $w\nequiv_M w'$, contradicting that $\mathfrak{P}$ is sound. 
Otherwise, either $\vert w\vert_{ab}\neq \vert w'\vert_{ab}$ or $\vert w\vert_{bc}\neq \vert w'\vert_{bc}$.
Assume $\vert w\vert_{ab}\neq \vert w'\vert_{ab}$.
Clearly, $ww'\Rightarrow_{R} w'w$  and $wcw'\Rightarrow_{R} w'cw$  because $(w,w')\in R$. By Lemma~\ref{1611a}, $\vert ww'\vert_u=\vert w'w\vert_u$ for $u\in \{ab,bc\}$. Hence, $ww' \Rightarrow_{\mathfrak{P}} w'w$ and since $\mathfrak{P}$ is sound, it follows that 
$\vert ww'\vert_{abc}= \vert w'w\vert_{abc}$.  On the other hand,
$\vert wcw'\vert_{ab}=\vert ww'\vert_{ab}= \vert w'w\vert_{ab}= \vert w'cw\vert_{ab}$ while
$\vert wcw'\vert_{bc}= \vert ww'\vert_{bc}+\vert w\vert_b= \vert w'w\vert_{bc}+\vert w'\vert_b= \vert w'cw\vert_{bc}$. Hence, $wcw' \Rightarrow_{\mathfrak{P}} w'cw$ as well.
However, $\vert wcw'\vert_{abc}= \vert ww'\vert_{abc}+\vert w\vert_{ab}
\neq \vert w'w\vert_{abc}+\vert w'\vert_{ab}=\vert w'cw\vert_{abc}$, contradicting that $\mathfrak{P}$ is sound.
The other case is similar.
\end{proof}

\begin{corollary}
Suppose $\Sigma=\{a<b<c\}$ and $\mathfrak{P}=(\Sigma, R,C)$ is a sound and complete Parikh rewriting system. If $R$ is finite, then $C=\{ab,bc,abc\}$.
\end{corollary}

\begin{proof}
Since $\mathfrak{P}$ is complete, by Remark~\ref{1001b}, $(\Sigma, R)$ is a Parikh complete Thue system.
By Theorem~\ref{1228d}, $(\Sigma, R)$ is not Parikh sound. Hence, by Theorem~\ref{2912a}, $abc\in C$.

Now, assume $ab\notin C$. Let 
$R'=     \{\, (y,y')\in R\mid y,y'\in \{a,b\}^*\,         \}$.  Since $(\Sigma,R)$ is Parikh complete, it can be easily deduced that $(\{a<b\}, R')$ is Parikh complete. 
If $w,w'\in \{a,b\}^*$ and $w\Rightarrow_{R'} w'$, then $w\Rightarrow_{\mathfrak{P}} w'$ because $ab\notin C$.
Since $\mathfrak{P}$ is sound, it follows that $(\{a<b\}, R')$ is Parikh sound as well.
By Theorem~\ref{1228d} again,
$R'$ is infinite and so is $R$, a contradiction. Hence, $ab\in C$. Similarly, it can be shown that $bc\in C$.
\end{proof}

\section{Irreducible Transformations}\label{0304a}

In this section, it will be shown that each transformation of a Parikh rewriting system
can be decomposed into a sequence of ``indecomposable" transformations. Before that, we need a particular notion of distance.\footnote{This is distinct from the notion of rank distance, denoted $d_R$, appeared in \cite{aA07,AAP08}.}

\begin{definition}
Suppose $(\Sigma,R)$ is a Thue system, $w,w'\in \Sigma^*$, and $w\Rightarrow_{R} w'$. The \emph{distance between $w$ and $w'$}, denoted $\dist_R(w,w')$,
is zero if $w=w'$; otherwise, it is the least positive integer $n$ such that there exist $w_0,w_1, \dotsc,w_n\in \Sigma^*$ satisfying
$w=w_0 \rightarrow_R w_1 \rightarrow_R w_2 \rightarrow_R \dotsb \rightarrow_R w_n=w'$.
\end{definition}

The distance function $\dist_R$ is indeed a metric on each equivalence class  $[w]_R$.  
When $R$ is understood, we would simply write $\dist(w,w')$.

\begin{definition}
Suppose $\mathfrak{P}=(\Sigma, R,C)$ is a Parikh rewriting system and suppose $w$ and $w'$ are distinct words over $\Sigma$.
We say that the transformation  $w \Rightarrow_{\mathfrak{P}}w'$ is \emph{reducible} 
if{f} there exists $w''\notin \{w,w'\}$ such that $w \Rightarrow_{\mathfrak{P}}w'' \Rightarrow_{\mathfrak{P}} w'$
and $\dist_R(w,w'')+\dist_R(w'',w')=\dist_R(w,w')$. Otherwise, we say that the transformation is \emph{irreducible} and denote this by $w\overset{\text{irr}}{\Rightarrow}_{\mathfrak{P}} w'$. The \emph{order} of an irreducible transformation $w\overset{\text{irr}}{\Rightarrow}_{\mathfrak{P}} w'$ is $n$ if{f} $\dist_R(w,w')=n$.
\end{definition}


\begin{remark}\label{1228c}
By definition, the identity transformations $w\Rightarrow_{\mathfrak{P}} w$ are not regarded as irreducible transformations. Let 
$$m=\min\{ \,\dist(w,w')\mid w,w'\in \Sigma^*, w \Rightarrow_{\mathfrak{P}} w' \text{ and } w\neq w'    \,  \}.$$
By Proposition~\ref{1501a}, $m$ is well-defined.
For every distinct $w,w'\in \Sigma^*$, if $w \Rightarrow_{\mathfrak{P}} w'$ and $\dist(w,w')=m$, then  
$w\overset{\text{irr}}{\Rightarrow}_{\mathfrak{P}} w'$. In particular, if $w\Rightarrow_{\mathfrak{P}} w'$ and $\dist(w,w')=1$, then $w\overset{\text{irr}}{\Rightarrow}_{\mathfrak{P}} w'$.
\end{remark}

It will be shown after Theorem~\ref{3112a} that $\overset{\textnormal{irr}}{\Rightarrow}_{\mathfrak{P}}$ and $\overset{\textnormal{irr}}{\Rightarrow}_{\mathfrak{Q}}$ need not be equal
even though $\mathfrak{P}$ and $\mathfrak{Q}$ are equivalent Parikh rewriting systems.

\begin{theorem}\label{2011a}
Suppose $\mathfrak{P}=(\Sigma, R,C)$ is a Parikh rewriting system and $w$ and $w'$ are distinct words over $\Sigma$. If $w\Rightarrow_{\mathfrak{P}} w'$, then
there exist $w_0, w_1, \dotsc, w_n\in \Sigma^*$ such that 
$$w=w_0 \overset{\textnormal{irr}}{\Rightarrow}_{\mathfrak{P}} w_1 \overset{\textnormal{irr}}{\Rightarrow}_{\mathfrak{P}} w_2 \overset{\textnormal{irr}}{\Rightarrow}_{\mathfrak{P}}\dotsb \overset{\textnormal{irr}}{\Rightarrow}_{\mathfrak{P}} w_n=w'$$
and $\dist_R(w,w')=\sum_{i=1}^n \dist_R(w_{i-1},w_i)$.
\end{theorem}

\begin{proof}
First of all, if $w\Rightarrow_{\mathfrak{P}} w'$ is irreducible, then the conclusion follows with \linebreak
$n=1$. We prove by induction on $\dist(w,w')$. If $\dist(w,w')$ equals one, then $w\Rightarrow_{\mathfrak{P}} w'$ is irreducible and we are done.
For the induction step, assume $w\Rightarrow_{\mathfrak{P}} w'$ is reducible.
Choose $w''\notin \{ w,w'\}$ such that $w \Rightarrow_{\mathfrak{P}}w'' \Rightarrow_{\mathfrak{P}} w'$
and $\dist(w,w'')+\dist(w'',w')=\dist(w,w')$. Clearly,  
$\dist(w,w'')<\dist(w,w')$ and $\dist(w'',w')<\dist(w,w')$. By the induction hypothesis, choose $u_0,u_1,\dotsc, u_p\in \Sigma^*$ such that 
$w=u_0 \overset{\textnormal{irr}}{\Rightarrow}_{\mathfrak{P}} u_1 \overset{\textnormal{irr}}{\Rightarrow}_{\mathfrak{P}} u_2 \overset{\textnormal{irr}}{\Rightarrow}_{\mathfrak{P}}\dotsb \overset{\textnormal{irr}}{\Rightarrow}_{\mathfrak{P}} u_p=w''$ and 
$\dist(w,w'')=\sum_{i=1}^p \dist(u_{i-1},u_i)$.
Similarly, choose $v_0,v_1,\dotsc, v_q\in \Sigma^*$ such that 
$w''=v_0 \overset{\textnormal{irr}}{\Rightarrow}_{\mathfrak{P}} v_1 \overset{\textnormal{irr}}{\Rightarrow}_{\mathfrak{P}} v_2 \overset{\textnormal{irr}}{\Rightarrow}_{\mathfrak{P}}\dotsb \overset{\textnormal{irr}}{\Rightarrow}_{\mathfrak{P}} v_q=w'$ and 
$\dist(w'',w')=\sum_{j=1}^q \dist(v_{j-1},v_j)$.
Therefore, 
$$w=u_0 \overset{\text{irr}}{\Rightarrow}_{\mathfrak{P}} u_1 \overset{\text{irr}}{\Rightarrow}_{\mathfrak{P}} \dotsb \overset{\text{irr}}{\Rightarrow}_{\mathfrak{P}} u_p=w''=v_0 \overset{\text{irr}}{\Rightarrow}_{\mathfrak{P}} v_1 \overset{\text{irr}}{\Rightarrow}_{\mathfrak{P}} \dotsb \overset{\text{irr}}{\Rightarrow}_{\mathfrak{P}}  v_q=w'$$
 and $\dist(w,w')=\sum_{i=1}^p \dist(u_{i-1},u_i)+\sum_{j=1}^q \dist(v_{j-1},v_j)$ as required.
\end{proof}

It is not surprising that the sequence of irreducible transformations guaranteed by Theorem~\ref{2011a} is not unique. 
However, the non-uniqueness of the length of that sequence is not so clear and will be adrressed later.



\begin{theorem}\label{1501c}
Suppose $\mathfrak{P}=(\Sigma, R,C)$ is a sound and complete Parikh rewriting system. 
Let $R'=\{\,(w,w') \mid w \overset{\text{irr}}{\Rightarrow}_{\mathfrak{P}}    w' \,  \}$.
Then $(\Sigma, R')$ is a Parikh sound and Parikh complete Thue system.
\end{theorem}

\begin{proof}
Suppose $w,w'\in \Sigma^*$. If $(w, w')\in R'$, then $w \Rightarrow_{\mathfrak{P}} w'$. Hence, $w\equiv_M w'$ because $\mathfrak{P}$ is sound.
Thus  $(\Sigma, R')$ is Parikh sound by Proposition~\ref{0403b}. 
Conversely, if $w\equiv_M w'$, then  $w \Rightarrow_{\mathfrak{P}} w'$ because $\mathfrak{P}$ is complete.
By Theorem~\ref{2011a}, it follows that $w\Rightarrow_{R'} w'$. Thus $(\Sigma, R')$ is Parikh complete.
\end{proof}

For some Parikh rewriting systems, the set of irreducible transformations may have a simple explicit description (see the very first theorem in the next section).
Therefore, Theorem~\ref{1501c} provides a plausible mean to obtain Parikh sound and Parikh complete Thue systems.

\section{The Binary Alphabet: a Case Study}\label{1703a}

In this section we study some Parikh rewriting systems for the binary alphabet. 
Suppose $\Sigma=\{a<b\}$ is a fixed ordered alphabet.


\begin{theorem}\label{3012a}
Suppose $\mathfrak{P}=(\Sigma, \{(ab,ba)\} ,\{ab\})$ and $(\Sigma, R')$ is the Thue system where $R'=\{ \, (abxba, baxab) \mid x\in \Sigma^*\,\}$.  Then for every $w,w'\in\Sigma^*$,
$$ w \overset{\textnormal{irr}}{\Rightarrow}_{\mathfrak{P}} w' \text{ if and only  if } w\rightarrow_{R'} w'.$$
\end{theorem}

\begin{proof}
Let $R=\{(ab,ba)\}$. By our definition, neither $w\overset{\textnormal{irr}}{\Rightarrow}_{\mathfrak{P}} w$ nor $w \rightarrow_{R'}w $  is true for each $w\in \Sigma^*$.
Suppose $w,w'\in\Sigma^*$ are distinct. First, we show that $w \overset{\textnormal{irr}}{\Rightarrow}_{\mathfrak{P}} w'$ if and only if 
 $w\Rightarrow_{\mathfrak{P}} w'$ and $\dist_R(w,w')=2$.
Clearly, if $w\Rightarrow_{\mathfrak{P}} w'$ and $\dist_R(w,w')=2$, then this transformation is irreducible by Remark~\ref{1228c}.

Conversely, assume $w \overset{\textnormal{irr}}{\Rightarrow}_{\mathfrak{P}} w'$ and so $w \Rightarrow_{\mathfrak{P}} w'$. Since $\vert w\vert_{ab}=\vert w'\vert_{ab}$, it follows that $\dist_R(w,w')$ must be even because the number of applications of the rewriting rule $ab\rightarrow ba$ must equal the number of applications of the rewriting rule $ba\rightarrow ab$. Assume $\dist_R(w,w')=2l $ for some $l>1$.
Suppose $w=w_0\rightarrow_R w_1 \rightarrow_R w_2\rightarrow_R \dotsb \rightarrow_R w_{2l} =w'$ for some $w_0,w_1, \dotsc, w_{2l}\in \Sigma^*$. 
Note that if  both the rewriting rules $ab\rightarrow ba$ and $ba\rightarrow ab$ are involved in some two consecutive direct transformations $w_i \rightarrow_R w_{i+1} \rightarrow_R w_{i+2} $, then the two rewriting rules must be applied to non-overlapping positions; otherwise, $w_i=w_{i+2}$ and $\dist_R(w,w')$ would have been less than $2l$. In this case, we may apply the rules in reverse order and obtain 
$w_i \rightarrow_R w_{i+1}' \rightarrow_R w_{i+2} $ for some $w_{i+1}'\in \Sigma^*$.
Therefore, commuting the applications of the rules $ab\rightarrow ba$ and $ba\rightarrow ab$ as many times as necessary,
we may assume that both rewriting rules are involved in the first two direct transformations $w_0 \rightarrow_R w_1 \rightarrow_R w_2 $.
However, this implies that $w\Rightarrow_{\mathfrak{P}} w_2 \Rightarrow_{\mathfrak{P}} w'$
and $\dist_R(w,w')=\dist_R(w,w_2)+\dist_R(w_2,w')$. Therefore,
$w \Rightarrow_{\mathfrak{P}} w'$ is reducible, a contradiction.

By now it should be clear that $w\rightarrow_{R'} w'$ if and only if    $w\Rightarrow_{\mathfrak{P}} w'$ and $\dist_R(w,w')=2$.
 Therefore, $w \overset{\textnormal{irr}}{\Rightarrow}_{\mathfrak{P}} w'$ if and only if $w\rightarrow_{R'} w'$ as required.
\end{proof}

Since $\mathfrak{P}=(\Sigma, \{(ab,ba)\} ,\{ab\})$ is a sound and complete Parikh rewriting system (see Example~\ref{1506a}), Theorem~\ref{1501c} and Theorem~\ref{3012a} together provide an alternative proof of the fact stated in Example~\ref{1506b}.
Furthermore, the proof of Theorem~\ref{3012a} shows that every irreducible transformation of $\mathfrak{P}$ has order two.


Although $(\Sigma, \{(ab,ba)\} ,\{ab\})$ is the most simple and natural Parikh rewriting system for the binary alphabet, it is intriguing
what other combination of rewriting rules such that the rule $ab\rightarrow ba$ is not already included may lead to. In view of Theorem~\ref{1501c}, our pursuit is accompanied by hope for a discovery of a new Parikh sound and Parikh complete Thue system for the binary alphabet.

For the remaining of this section, fix $R_1=\{(abb, bab), (bab, bba), (bba, abb)\}$
and $R_2= \{(baa, aba), (aba, aab), (aab, baa)\}$.

For each $i=1,2$, it is easy to verify that $abba\nRightarrow_{R_i} baab$. Since  $abba \equiv_M baab$, it follows that neither $(\Sigma, R_1)$ nor $(\Sigma, R_2)$  is a Parikh complete Thue system.

\begin{theorem}\label{3112a}
 If $R\subseteq R_1\cup R_2$, $R\cap R_i\neq \emptyset$ for each $i=1,2$, and $\vert R\cap R_i\vert \geq 2$ for some $i=1,2$, then $\mathfrak{P}=( \Sigma, R, \{ab\})$ is 
 a sound and complete Parikh rewriting system.
\end{theorem}

\begin{proof}
Suppose $w,w'\in \Sigma^*$. We need to show $w \Rightarrow_{\mathfrak{P}} w'$ if and only if $w\equiv_M w'$.
If $w \Rightarrow_{\mathfrak{P}} w'$, then $\vert w\vert_{ab}=\vert w'\vert_{ab}$. Since every rewriting rule of $\mathfrak{P}$ preserves the Parikh vector, it follows that $w\equiv_M w'$. Conversely, assume $w\equiv_M w'$. To see that $w \Rightarrow_{\mathfrak{P}} w'$, since $(\Sigma,R')$ is Parikh complete,  where 
$R'=\{ ( abxba, baxab) \mid x\in \Sigma^*\}$, it suffices to show that $\rightarrow_{R'}\subseteq \Rightarrow_{\mathfrak{P}}$.

Therefore, it remains to show that $abwba \Rightarrow_{\mathfrak{P}} bawab$ for every $w\in \Sigma^*$. 
Fix arbitrary $w\in \Sigma^*$. Since $\vert abwba\vert_{ab}=\vert bawab\vert_{ab}$, it suffices to show that $abwba \Rightarrow_R bawab$.
We may assume $\vert R\cap R_1\vert \geq 2$ as the other case mirrors this.
Furthermore, any rewriting rule induced by an ordered pair in $R_1$ produces the same effect as a combination of rewriting rules induced by the other two ordered pairs in $R_1$.
Hence, without loss of generality, we may further assume that $R_1\subseteq R$.

\begin{itemize}
\item[Case 1:] $R \supseteq R_1\cup \{(aba, aab)\}$.\\
Note that $bawba\rightarrow_R bawab$: if $w=\lambda$ or $w[ \vert w\vert]=a$, apply the (rewriting) rule
$aba\rightarrow aab$; otherwise, apply the rule $bba\rightarrow bab$. 
Hence, it remains to show that $abwba \Rightarrow_R bawba$.
Now, if $w=\lambda$ or $w[1]=b$, then $abwba \rightarrow_R bawba$, applying the rule $abb\rightarrow bab$.
Otherwise, if $w[1]=a$, write $wb$ as $aw'bw''$ for some $w'\in a^*$ and $w''\in \Sigma^*$.
Using the rule $aab\rightarrow aba$ as many times as necessary (possibly none), it follows that $abwba=abaw'bw''a \Rightarrow_R ababw'w''a$.
Then $ababw'w''a \rightarrow_R abbaw'w''a\rightarrow_R babaw'w''a$, using the rule $bab\rightarrow bba$, followed by the rule $abb \rightarrow bab$. Finally, using the rule $aba\rightarrow aab$ as many times as necessary,
it follows that $babaw'w''a\Rightarrow_R baaw'bw''a=bawba$.

\item[Case 2:] $R\supseteq R_1\cup \{(aab, baa)\}$.\\
First, we show that $abwba \Rightarrow_R bawba$. If $w=\lambda$ or $w[1]=b$, then apply the rule $abb\rightarrow bab$.
Otherwise, if $w[1]=a$, write $wb$ as $aw'bw''$ for some $w'\in a^*$ and $w''\in \Sigma^*$.
Using the rule $aab\rightarrow baa$ as many times as necessary and then possibly the rule $bab\rightarrow bba$ once (if $\vert w'\vert$ is even),
it follows that $abwba=abaw'bw''a \Rightarrow_R abbaw'w''a$. Then $abbaw'w''a \rightarrow_R bbaaw'w''a$, applying the rule $abb\rightarrow bba$. Using possibly the rule $bba\rightarrow bab$ once (if $\vert w'\vert$ is odd) and then the rule $baa\rightarrow aab$ as many times as necessary,
it follows that $bbaaw'w''a\Rightarrow_R baaw'bw''a=bawba$.

Similarly, it can be shown that $bawba \Rightarrow_R bawab$.
If $w[\vert w\vert]=b$, then $bawba \rightarrow_R bawab$, applying the rule $bba\rightarrow bab$.
Otherwise, if $w=\lambda$ or $w[\vert w\vert]=a$, then write $baw$ as $w'baw''$ for some $w'\in \Sigma^*$ and $w''\in a^*$.
Using the rule $aab\rightarrow baa$ as many times as necessary and then possibly the rule $bab\rightarrow bba$ once (if $\vert w''\vert$ is even),
it follows that $bawba=w'baw''ba \Rightarrow_R w'bbaw''a$. Then using possibly the rule $bba\rightarrow bab$ once (if $\vert w'\vert$ is odd) and then the rule $baa\rightarrow aab$ as many times as necessary, it follows that $w'bbaw''a\Rightarrow_R w'baw''ab=bawab$.

\item[Case 3:] $R \supseteq R_1\cup \{(baa, aba)\}$.\\
This case is similar to Case~1.  \qedhere
\end{itemize}
\end{proof}

\begin{remark}
If $\vert R\cap R_i\vert =1$ for both $i=1,2$, then $(\Sigma, R)$ is not a Parikh complete Thue system. Hence, $(\Sigma,R,\{ab\})$ is not a complete Parikh rewriting system.
\end{remark}

Suppose $\mathfrak{P}=(\Sigma, R, \{ab\})$, where $ R=R_1\cup R_2$,  and $\mathfrak{Q}=(\Sigma,\{(ab,ba)\}, \{ab\})$.
Then $\mathfrak{P}$ and $\mathfrak{Q}$ are equivalent as they are both sound and complete. 
However, we will show that $\overset{\textnormal{irr}}{\Rightarrow}_{\mathfrak{P}}\neq \overset{\textnormal{irr}}{\Rightarrow}_{\mathfrak{Q}}$. 
Clearly, $w= \boldsymbol{bba}aabaab \rightarrow_R    abb\boldsymbol{aab} aab \rightarrow_R abbaba\boldsymbol{aab}\rightarrow_R    abbabaaba= w' $. Since $w$ and $w'$ differ in six positions and every rewriting rule of $\mathfrak{P}$ affects exactly two positions, it follows that $\dist_{R}(w,w')=3$. Hence, $w\Rightarrow_{\mathfrak{P}} w'$ must be irreducible because no (irreducible) transformation of $\mathfrak{P}$ has order one.
However, by Theorem~\ref{3012a}, $w\overset{\text{irr}}{\nRightarrow}_{\mathfrak{Q}} w'$.

Although $w\overset{\text{irr}}{\Rightarrow}_{\mathfrak{P}} w'$ has order three, $w$ can be transformed into $w'$ via a sequence of irreducible
transformations of order two, namely 
$$w=\boldsymbol{bba}\mathbf{aab}aab \overset{\text{irr}}{\Rightarrow}_{\mathfrak{P}} abb\boldsymbol{baa}\mathbf{aab} \overset{\text{irr}}{\Rightarrow}_{\mathfrak{P}}abbabaaba=w'.$$
In fact, since $(\Sigma, R')$ is Parikh complete, where
$R'=\{\,  (abxba,baxab)\mid x\in \Sigma^*\,\}$,
 and every direct transformation $w\rightarrow_{R'} w'$ is an irreducible transformation of $\mathfrak{P}$ of order two, it implies that every non-identity transformation of $\mathfrak{P}$ can be expressed as a sequence of irreducible transformations of order two. Therefore, it raises the question whether the  irreducible transformations of order at most two of a complete Parikh rewriting system  are sufficient to form a Parikh complete Thue system. This will be answered in the next section.


\section{Salomaa's Parikh Rewriting System}\label{2703a}

For this section, suppose $\Sigma=\{a<b<c\}$ is a fixed ordered alphabet. Let $\mathfrak{P}=( \Sigma,   R   ,\{abc\}     )$ be a fixed Parikh rewriting system, where
$$R=\{ (ac, ca)\}  \cup \{ \, (abxba, baxab) \mid x\in \Sigma^*\,\}\cup\{\, (bcxcb,cbxbc) \mid x\in \Sigma^*   \, \}. $$
 Recall that $\mathfrak{P}$ is  sound and complete.



Since $abc$ is the unique counter for $\mathfrak{P}$, when we write $w\overset{+3}{\rightarrow}_R w'$, the number above $\rightarrow_R$ indicates the corresponding change in the number of occurrences of the subword $abc$, that is, $\vert w'\vert_{abc}=\vert w\vert_{abc}+3$. Clearly, 
if $w=w_0\overset{t_1}{\rightarrow}_R w_1 \overset{t_2}{\rightarrow}_R w_2 \overset{t_3}{\rightarrow}_R   \dotsb \overset{t_n}{\rightarrow}_R w_n=w'$
and $w \Rightarrow_{\mathfrak{P}} w'$ (so $\vert w\vert_{abc}=\vert w'\vert_{abc}$), then $\sum_{i=1}^n t_i=0$ must be zero.

\begin{remark}\label{0903a}
Note that each rewriting rule of $\mathfrak{P}$ affects at most four letters.
Hence, if $w,w'\in \Sigma^*$ and $\dist(w,w')=n$, then $w, w'$ can differ in at most $4n$ positions. 
Conversely, if $w, w'$ differ in $4n$ positions and
$w=w_0\rightarrow_R w_1 \rightarrow_R w_2 \rightarrow_R \dotsb \rightarrow_R w_n=w'$ for some $w_0,w_1, \dotsc, w_n \in \Sigma^*$, then 
$\dist(w,w')=n$. 
\end{remark}

\begin{example}
Consider the following sequence of direct transformations. The affected letters in each step are highlighted in bold.
\begin{align*}
w=abc\boldsymbol{ba}bacababcbabac\boldsymbol{ab} \overset{+3}{\rightarrow}_R & \,\boldsymbol{ab}cab\boldsymbol{ba}cababcbabacba\\
\overset{-1}{\rightarrow}_R &\, bacababc    \boldsymbol{ab}    abc \boldsymbol{ba}  ba   cba \\
\overset{-1}{\rightarrow}_R &\, bacababc   ba   \boldsymbol{ab}          c  ab   \boldsymbol{ba}           cba\\
\overset{-1}{\rightarrow}_R &\, bacababcba  ba   c   ab ab   cba=w'
\end{align*}
By Remark~\ref{0903a}, $\dist(w,w')=4$ as $w$ and $w'$ differ in 16 positions.
Note that no proper combination of the values above the arrows sums to zero. However, $w \Rightarrow_{\mathfrak{P}}w'$ is reducible 
because
\begin{align*}
w=\boldsymbol{ab}c\boldsymbol{ba}\mathbf{ba}c\mathbf{ab}abcbabacab \Rightarrow_{\mathfrak{P}} & \, bacababcba\boldsymbol{ab}c  \boldsymbol{ba}      \mathbf{ba}  c \mathbf{ab}=w''\\
   \Rightarrow_{\mathfrak{P}} & \, bacababcba  ba    c  ab  ab   c ba   =w'
\end{align*}
 and $\dist(w,w'')+\dist(w'',w')=2+2=\dist(w,w')$.  
\end{example}



The following technical lemma will be useful in determining $\dist_R(w,w')$ in more complicated scenarios.

\begin{lemma}\label{1303a}
Suppose $ R_1=\{ \, (abxba, baxab) \mid x\in \Sigma^*\,     \}$, $R_2= \{\, (bcxcb, cbxbc) \mid x\in \Sigma^*\, \}$, and  $R_3=\{(ac,ca)\}$.
If $w,w'\in \Sigma^*$ and  $w\Rightarrow_{R} w'$, then
\begin{enumerate}
\item $\pi_{a,b}(w)\Rightarrow_{R_1} \pi_{a,b}(w')$, $\pi_{b,c}(w)\Rightarrow_{R_2} \pi_{b,c}(w')$, and $\pi_{a,c}(w)\Rightarrow_{R_3} \pi_{a,c}(w')$;
\smallskip
\item $\dist_R(w,w')\geq  \dist_{R_1}( \pi_{a,b}(w),\pi_{a,b}(w'))+\dist_{R_2}(  \pi_{b,c}(w), \pi_{b,c}(w'))$\\
\phantom{$\dist_R(w,w')\geq   $   }$+\dist_{R_3}(  \pi_{a,c}(w), \pi_{a,c}(w'))$.
\end{enumerate}
\end{lemma}

\begin{proof}
Suppose $w=w_0 \rightarrow_R w_1 \rightarrow_R w_2 \rightarrow_R \dotsb \rightarrow_R w_{\dist_R(w,w')}=w'$ for some $w_0,w_1, \dotsc, w_{ \dist_R(w,w')  }\in \Sigma^*$.
Notice that exactly  the  rewriting rules induced by elements of $R_1$ would affect the relative positions of the $a$'s and $b$'s in  $w$. 
Suppose $i_1 ,i_2, \dotsc,i_p$ is an increasing enumeration of all the $i$'s such that the direct transformation $w_{i-1} \rightarrow_R w_{i}$ is carried out using a rewriting rule induced by an element of $R_1$. If no such $i$ exists, then $\pi_{a,b}(w)= \pi_{a,b}(w')$ and by definition $\pi_{a,b}(w)\Rightarrow_{R_1} \pi_{a,b}(w')$. Otherwise, it should be clear that 
$$\pi_{a,b}(w) \rightarrow_{R_1} \pi_{a,b}(w_{i_1}) \rightarrow_{R_1} \pi_{a,b}(w_{i_2}) \rightarrow_{R_1}\dotsb \rightarrow_{R_1} \pi_{a,b}(w_{i_p})= \pi_{a,b}(w').$$
Therefore, $\pi_{a,b}(w)\Rightarrow_{R_1} \pi_{a,b}(w')$ and $p \geq \dist_{R_1}( \pi_{a,b}(w),\pi_{a,b}(w'))$. 
Similarly, $\pi_{b,c}(w)\Rightarrow_{R_2} \pi_{b,c}(w')$ and $\pi_{a,c}(w)\Rightarrow_{R_3} \pi_{a,c}(w')$. The second part follows from the proof of the first part as $R_1$, $R_2$, and $R_3$ are disjoint.
\end{proof}

\begin{remark}
The conclusion stated in Lemma~\ref{1303a} (2) cannot be strengthen to equality, not even assuming $w\equiv_M w'$. For example,
let $w=bcacabcabbca$ and $w'=cabbcabcacab$. Note that $w\Rightarrow_R w'$ and $w\equiv_M w'$.  However, only the rewriting rules $ac\rightarrow ca$ or $ca\rightarrow ac$ can be applied on $w$ but $\pi_{a,c}(w)=\pi_{a,c}(w')=cacacaca$. Therefore, from the proof of Lemma~\ref{1303a},  we see that $\dist_R(w,w')$ is strictly greater than the sum of the  corresponding distances on the right. 
\end{remark}

\begin{example}
Since $w=ab\boldsymbol{bc}a\boldsymbol{cb} \overset{+1}{\rightarrow}_R \boldsymbol{ab}c\boldsymbol{ba}bc \overset{-1}{\rightarrow}_R   bacabbc=w'$,
the transformation $w\Rightarrow_{\mathfrak{P}}w'$ is irreducible with order 2.
In this case, the two words happen to be \mbox{$ME$-equivalent}, as 
$$w=abb\boldsymbol{ca}cb \overset{0}{\rightarrow}_R \boldsymbol{abba}ccb \overset{0}{\rightarrow}_R baa\boldsymbol{bccb} \overset{0}{\rightarrow}_R ba\boldsymbol{ac}bbc=w'' \overset{0}{\rightarrow}_R bacabbc=w'.$$
Using this, we can address the uniqueness issue that was raised after Theorem~\ref{2011a}. By Lemma~\ref{1303a},  $\dist( w,w'')\geq 3$ and hence it must be three. However, $w \Rightarrow_{\mathfrak{P}} w''$ can be expressed as two sequences of irreducible transformations (of distinct length) satisfying the conclusion of Theorem~\ref{2011a}, namely, $w \overset{\text{irr}}{\Rightarrow}_{\mathfrak{P}} w' \overset{\text{irr}}{\Rightarrow}_{\mathfrak{P}} w'' $
and  $w  \overset{\text{irr}}{\Rightarrow}_{\mathfrak{P}} abbaccb \overset{\text{irr}}{\Rightarrow}_{\mathfrak{P}} baabccb  \overset{\text{irr}}{\Rightarrow}_{\mathfrak{P}} w''$. 
\end{example}

\begin{theorem}\label{0301a}
For every positive integer $n$, there exist $w,w'\in \Sigma^*$ such that $w \Rightarrow_{\mathfrak{P} }w'$ is irreducible with order $n$.
\end{theorem}

\begin{proof}
The transformation $ac\Rightarrow_{\mathfrak{P}} ca$ is irreducible with order one.
For every positive integer $n$, consider the following sequence of direct transformations:
\begin{multline*}
w=  a^nbcba^nac^nab \overset{-1}{\rightarrow}_R a^{n-1}bacaba^{n-1}ac^nab\overset{-1}{\rightarrow}_R \dotsb\\
\overset{-1}{\rightarrow}_R ba^nca^nbac^nab \overset{+n}{\rightarrow}_R  ba^nca^nabc^nba  = w'.
\end{multline*}
Since the number of $a$ between the first two $b$ in $w$ is zero while that number is $2n+1$ in $w'$, it follows that 
$\dist(w,w')=n+1$. In fact, if $w=w_0 \overset{i_1}{\rightarrow}_R w_1 \overset{i_2}{\rightarrow}_R \dotsb \overset{i_{n+1}}{\rightarrow}_R w_{n+1}=w'$
for some $w_0,w_1,\dotsc, w_{n+1}\in \Sigma^*$, then
\begin{enumerate}
\item each except one of the direct transformations applies the rewriting rule of the form $abxba\rightarrow baxab$ for some $x\in \Sigma^*$ such that the first and the second $b$ in the corresponding word are affected;
\item exactly one of the direct transformations applies the rewriting rule
of the form $baxab\rightarrow abxba$ for some $x\in \Sigma^*$ such that the second and the third $b$ in the corresponding word are affected.
\end{enumerate}
This means that every $i_k$ is $-1$ except one that is $+n$.
Assume $w\Rightarrow_{\mathfrak{P}} w'$ is reducible. Then choose $w'' \notin \{w,w'\}$ 
such that $w\Rightarrow_{\mathfrak{P}} w''\Rightarrow_{\mathfrak{P}} w'$ and $\dist(w,w'')+\dist(w'',w')=n+1$.
Let $d=\dist(w,w'')$. Then 
$$w=w_0 \overset{i_1}{\rightarrow}_R w_1 \overset{i_2}{\rightarrow}_R \dotsb \overset{i_d}{\rightarrow}_R w_d=w''\overset{ i_{d+1}}{\rightarrow}_R\dotsc   \overset{i_{n+1}}{\rightarrow}_R w_{n+1}=w'$$
for some $w_0, w_1, \dotsc, w_{n+1}\in \Sigma^*$ such that $i_1+i_2+\dotsb+i_d=0$. However, by the observation above, 
$i_1+i_2+\dotsb+i_d$ is either $-d$ or $n+1-d$, a contradiction.
\end{proof}

Although Theorem~\ref{0301a} shows the existence of irreducible transformations of arbitrarily large order, 
as pointed out at the end of Section~\ref{1703a}, it remains to be seen whether irreducible transformations of large order are indeed indispensable.
Our next theorem partially answers this question positively.

\begin{theorem}\label{0701f}
Irreducible transformations of $\mathfrak{P}$ of order three
that cannot be expressed as a  sequence of irreducible transformations of order at most two exist.
\end{theorem}

\begin{proof}
(The readers may wish to verify the result on their own, perhaps aided by a computer, rather than following our detailed case by case proof here.) By the proof of Theorem~\ref{0301a}, $aabcbaaaccab\Rightarrow_{\mathfrak{P}} baacaaabccba$ is an irreducible transformation of order three.
We claim that it cannot be expressed as a sequence of irreducible transformations of order at most two. 
First of all, suppose $w$ is obtained from $aabcbaaaccab$ using a sequence of irreducible transformations of order one. 
Then it should be clear that $w$ must be $aabcbub$ for some $u\in \Sigma^*$ such that $\vert u\vert_a=4$, $\vert u\vert_b=0$, and $\vert u\vert_c=2$.
For easier visualization and brevity, we simply represent this by 
$$w=aabcb\underbrace{u}_{4a, 2c}b.$$
It suffices to show that no irreducible transformation of order two can be applied to $w$. We argue by contradiction.
Assume $w\Rightarrow_{\mathfrak{P}} w'$ is irreducible with order two for some $w'\in \Sigma^*$. Then for some $w''\notin\{w,w'\}$,
$$w  \overset{p}{\rightarrow}_R w'' \overset{q}{\rightarrow}_R w', \text{ where } 0\neq p=-q.$$
For the purpose of referencing, let us refer to the rewriting rule employed by $w  \overset{p}{\rightarrow}_R w''$ (respectively $w''  \overset{q}{\rightarrow}_R w'$) as Rule I (respectively Rule II).  
Note that if Rule I is either $abxba\rightarrow baxab$ or $cbxbc\rightarrow bcxcb$ for some $x\in \Sigma^*$, then Rule II must be either $bayab\rightarrow abyba$ or $bcycb\rightarrow cbybc$ for some $y\in \Sigma^*$ and vice versa.

\begin{itemize}
\item[Case 1.] Rule I is $abxba\rightarrow baxab$ for some $x\in \Sigma^*$.\\
Then $w = a\boldsymbol{ab}c\boldsymbol{ba} \underbrace{v}_{3a, 2c}  b \overset{-1}{\rightarrow}_R  abacab \underbrace{v}_{3a, 2c}   b=w''$.
If Rule II is $bcycb\rightarrow cbybc$ for some $y\in \Sigma^*$.
Then 
$ w''=  abaca\boldsymbol{bc}aaa\boldsymbol{cb}  \overset{+3}{\rightarrow}_R  abacacbaaabc=w'$.
However, $p+q=-1+3\neq 0$, a contradiction. 
Therefore, Rule II must be  $bayab\rightarrow abyba$ for some $y\in \Sigma^*$.
There are two possible applications of this rule on $w''$ that do not revert $w''$ back to $w$.
Since $\vert v\vert_c=2$, these two applications would increase the counter by at least two. 
Hence, again $p+q\neq 0$. 

\item[Case 2.]  Rule I is $cbxbc\rightarrow bcxcb$ for some $x\in \Sigma^*$.\\
This case cannot happen as this rule cannot be applied on $w$.

\item[Case 3.] Rule I is $baxab\rightarrow abxba$ for some $x\in \Sigma^*$.\\
Then $w = aabc\boldsymbol{ba} \underbrace{v}_{2a, 2c}  \boldsymbol{ab} \overset{+2}{\rightarrow}_R  aabcab \underbrace{v}_{2a, 2c} ba=w''$.
Then Rule II must be $abyba\rightarrow bayab$ for some $y\in \Sigma^*$ as the rule
$cbybc\rightarrow bcycb$ for any $y\in \Sigma^*$ cannot be applied on $w''$.
There are two possible applications of this rule on $w''$ that do not revert $w''$ back to $w$.
These two applications would decrease the counter by  1 and 3 respectively. 
Hence, again $p+q\neq 0$.

\item[Case 4.] Rule I is $bcxcb\rightarrow cbxbc$ for some $x\in \Sigma^*$.\\
There are two possible distinct applications of Rule I on $w$. 
\begin{itemize}
\item $w= aa\boldsymbol{bc}b \underbrace{v}_{4a, 1c}    \boldsymbol{cb} \overset{+4}{\rightarrow}_R  aacbb \underbrace{v}_{4a, 1c} bc=w''$.
\item $w = aabc\boldsymbol{bc}aaaa\boldsymbol{cb} \overset{+4}{\rightarrow}_R  aabccbaaaabc=w''$.
\end{itemize}
If $w''=aacbb \underbrace{v}_{4a, 1c}bc$, then Rule II  must be $cbybc\rightarrow bcycb$ for some $y\in \Sigma^*$
as the rule $abyba\rightarrow bayab$ for any $y\in \Sigma^*$ cannot be applied on $w''$. 
However, this rule cannot account for $w'' \overset{-4}{\rightarrow}_R w'$, unless $w'=w$, which is not the case.
On the other hand, if $w''=aabccbaaaabc$, then \mbox{Rule II} must be $abyba\rightarrow bayab$ for some $y\in \Sigma^*$ to avoid $w''$ being reverted back to $w$.
Thus $w''=a\boldsymbol{ab}cc\boldsymbol{ba}aaabc \overset{-2}{\rightarrow}_R   abaccabaaabc=w'$. Nevertheless, again $p+q\neq 0$.
\qedhere
\end{itemize}
\end{proof}

It should be pointed out that irreducible transformations of $\mathfrak{P}$ of order three that can be expressed as a sequence of irreducible transformations of order at most two do exist.
Without going through the similar tedious arguments, we claim that
$abcbcbacab\Rightarrow_{\mathfrak{P}} bacabcbcba$ is irreducible with order three.
On the other hand, $abcbcba\boldsymbol{ca}b \Rightarrow_{\mathfrak{P}} \boldsymbol{ab}c\mathbf{bc}\boldsymbol{ba}a\mathbf{cb} \Rightarrow_{\mathfrak{P}} bacc\boldsymbol{ba}b\boldsymbol{ab}c \Rightarrow_{\mathfrak{P}} bac\boldsymbol{ca}bbbac \Rightarrow_{\mathfrak{P}} bacacbbb\boldsymbol{ac} \Rightarrow_{\mathfrak{P}} baca\boldsymbol{cb}b\boldsymbol{bc}a
\Rightarrow_{\mathfrak{P}} bacabcbcba$ is a sequence of irreducible transformations of order at most two.

\section{Conclusions}

Parikh rewriting systems are feasible alternatives to Thue systems in the quest for characterization of $M$-equivalence.
As highlighted by Theorem~\ref{1501c}, the irreducible transformations of a sound and complete Parikh rewriting system leads to a Parikh sound and Parikh complete Thue system. It is conceivable that with an ingeniously chosen Parikh rewriting system with a decidable set of rewriting rules, the set of \emph{all} irreducible transformations may have an effective description. While our immediate goal is the ternary alphabet,
this approach may prove to be a universal way of generating Parikh sound and Parikh complete Thue systems for every alphabet.

For the ternary alphabet, the  only Parikh rewriting system $\mathfrak{P}$ that we have analyzed in the previous section is chosen due to its appearance in the literature and its canonicalness. It is possible to give an exhaustive description of all its irreducible transformations of order two or three. However, this task is not carried out explicitly as we believe that not every irreducible transformation of $\mathfrak{P}$ can be expressed
as a sequence of irreducible transformations of order at most three anyway. As the proof of Theorem~\ref{0701f} is by a tedious case analysis, its proof cannot be adapted to analogous results
for higher order. Hence, the following is left as a conjecture.

\begin{conjecture}
Irreducible transformations of $\mathfrak{P}$ of arbitrary large order that cannot be expressed as a sequence of irreduccible transformations of lower order exist.
\end{conjecture}


\section*{Acknowledgements}

The author gratefully acknowledge support for this research by a short term grant No. 304/PMATHS/6313077 of Universiti Sains Malaysia.


\end{document}